\documentclass[conference]{IEEEtran}
\IEEEoverridecommandlockouts

\usepackage{cite}
\usepackage{amsmath,amssymb,amsfonts}
\usepackage{graphicx}
\usepackage{textcomp}
\usepackage{hyperref}       
\usepackage{url}            
\usepackage{booktabs}       
\usepackage{amsfonts}       
\usepackage{nicefrac}       
\usepackage{microtype}            
\usepackage{enumitem}
\usepackage{amsmath}
\usepackage{amsthm}
\usepackage{graphicx}
\usepackage[table]{xcolor}
\usepackage{multirow}
\usepackage{multicol}
\usepackage{svg}
\usepackage{wrapfig}
\usepackage[numbers]{natbib}
\usepackage[ruled,vlined,linesnumbered]{algorithm2e}
\usepackage{algpseudocode}
\usepackage{cleveref}
\usepackage{tikz}
\usepackage{balance}

\definecolor{darkred}{rgb}{0.5, 0.0, 0.0}
\definecolor{skyblue}{RGB}{203, 221, 245}
\newcommand{\skyblue}{\rowcolor{skyblue}}

\newif\ifarxiv
\arxivfalse

\usepackage[most]{tcolorbox}
\tcbset{
  aibox/.style={
    width=474.18663pt,
    top=10pt,
    colback=white,
    colframe=black,
    colbacktitle=black,
    enhanced,
    center,
    attach boxed title to top left={yshift=-0.1in,xshift=0.15in},
    boxed title style={boxrule=0pt,colframe=white,},
  }
}
\newtcolorbox{AIbox}[2][]{aibox,title=#2,#1}

\newtheorem{definition}{Definition}
\newtheorem{theorem}{Theorem}

\def\BibTeX{{\rm B\kern-.05em{\sc i\kern-.025em b}\kern-.08em
    T\kern-.1667em\lower.7ex\hbox{E}\kern-.125emX}}
\begin{document}

\title{EraRAG: Efficient and Incremental Retrieval Augmented Generation for Growing Corpora\\
}

\author{
    \IEEEauthorblockN{$^1$Fangyuan Zhang\IEEEauthorrefmark{1}\thanks{\IEEEauthorrefmark{1} ALL authors contributed equally to this research.},
                          $^2$Zhengjun Huang\IEEEauthorrefmark{1},
                          $^3$Yingli Zhou\IEEEauthorrefmark{1}\IEEEauthorrefmark{2}\thanks{\IEEEauthorrefmark{2} Yingli Zhou is the corresponding author.},
                      $^2$Qintian Guo,
 					  $^4$Zhixun Li,
                      $^5$Wensheng Luo,\\
                      $^6$Di Jiang,
                      $^3$Yixiang Fang,
                      $^2$Xiaofang Zhou,~\IEEEmembership{Fellow,~IEEE}}
    \IEEEauthorblockA{$^1$Huawei Hong Kong Research Center, Hong Kong; $^2$The Hong Kong University of Science and Technology, Hong Kong;\\
    $^3$The Chinese University of Hong Kong-Shenzhen, Shenzhen; $^4$The Chinese University of Hong Kong, Hong Kong;\\
    $^5$Hunan University, Changsha; $^6$WeBank, Shenzhen
    }
    \IEEEauthorblockA{zhang.fangyuan@huawei.com; zhuangff@connect.ust.hk; yinglizhou@link.cuhk.edu.cn; qtguo@ust.hk; \\ zxli@se.cuhk.edu.hk; luowensheng@hnu.edu.cn; dijiang@webank.com; fangyixiang@cuhk.edu.cn; zxf@cse.ust.hk}
}






\maketitle

\begin{abstract}
  Graph-based Retrieval-Augmented Generation (Graph-RAG) enhances large language models (LLMs) by structuring retrieval over an external corpus. However, existing approaches typically assume a static corpus, requiring expensive full-graph reconstruction whenever new documents arrive, limiting their scalability in dynamic, evolving environments. To address these limitations, we introduce EraRAG, a novel multi-layered Graph-RAG framework that supports efficient and scalable dynamic updates. Our method leverages hyperplane-based Locality-Sensitive Hashing (LSH) to partition and organize the original corpus into hierarchical graph structures, enabling efficient and localized insertions of new data without disrupting the existing topology. The design eliminates the need for retraining or costly recomputation while preserving high retrieval accuracy and low latency. Experiments on large-scale benchmarks demonstrate that EraRag achieves up to an order of magnitude reduction in update time and token consumption compared to existing Graph-RAG systems, while providing superior accuracy performance. This work offers a practical path forward for RAG systems that must operate over continually growing corpora, bridging the gap between retrieval efficiency and adaptability. Our code and data are available at \url{https://github.com/EverM0re/EraRAG-Official}.
\end{abstract}

\section{Introduction}

The emergence of Large Language Models (LLMs) such as GPT-4~\cite{gpt4}, Qwen~\cite{qwen}, and LLaMA~\cite{llama} has advanced natural language processing, achieving state-of-the-art results across various tasks~\cite{summary,code,sci01,li2023survey}. Despite their scalability and generalization, LLMs still struggle with domain-specific queries, multi-hop reasoning, and deep contextual understanding~\cite{domain1,domain2}, often yielding incorrect or hallucinated outputs~\cite{hallu,hallu1,hallu2} due to gaps in domain or real-time knowledge within their pretraining corpus. Fine-tuning with domain data~\cite{tune} can help but is often costly and yields limited gains in low-resource settings~\cite{limi1,limi2}. To address these limitations, Retrieval-Augmented Generation (RAG)~\cite{rag1,rag2,rag3,rag4,rag5} has emerged as a compelling approach, enriching LLMs with external knowledge to enhance factuality, interpret ability, and trust~\cite{trust,trust1,trust2,trust3,reason}. RAG retrieves relevant content from text corpora, structured datasets, or knowledge graphs to support tasks such as answering questions. Recent work has emphasized graph-structured memory, enabling richer semantic representation and multi-hop reasoning~\cite{gr1,gr2,gr3,gr4,gr5}.

\begin{figure}[t]
  \centering
  \includegraphics[width=\linewidth]{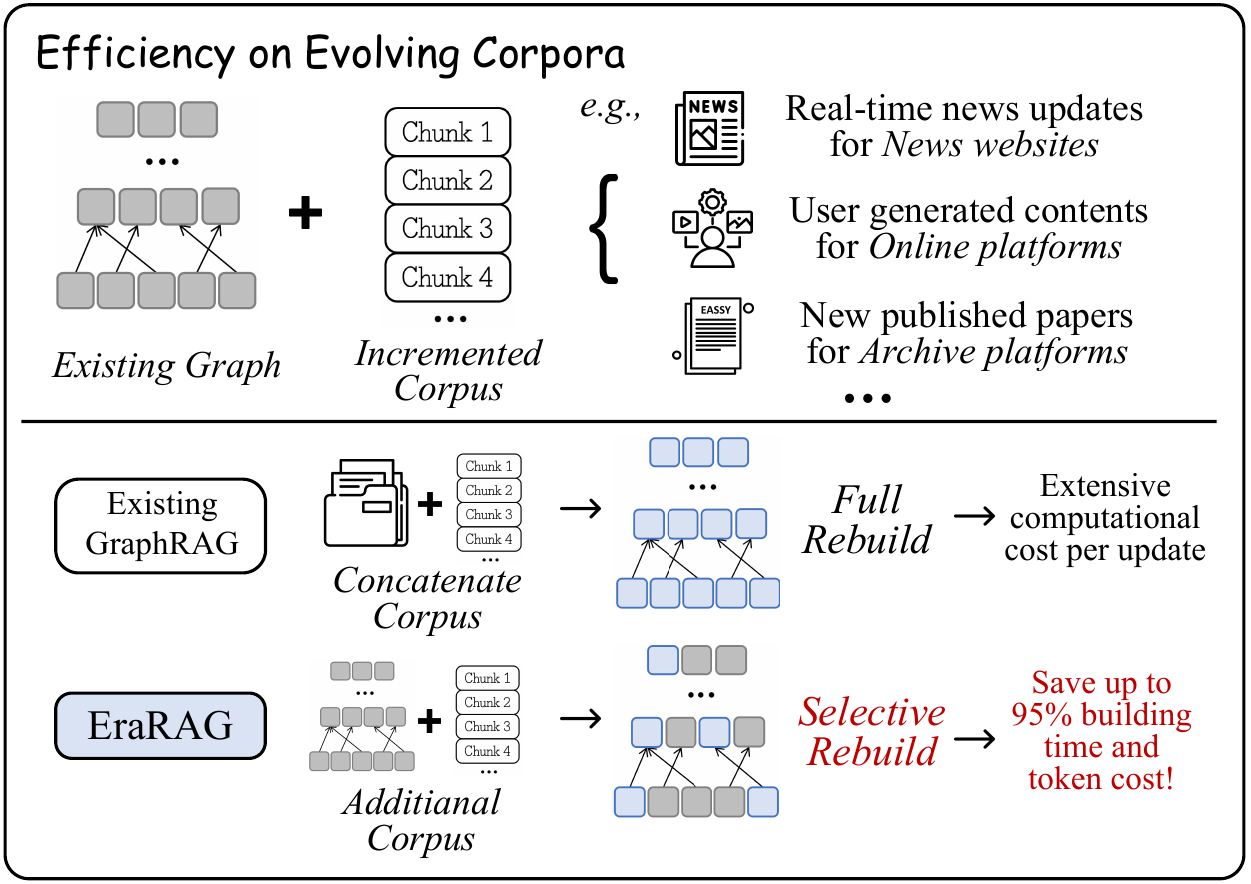}
  \caption{An illustrative example demonstrating the limitations of existing RAG methods and the advantages of \texttt{EraRAG}.}
  \label{fig:example}
\end{figure}

\begin{figure*}[t]
  \centering
  \includegraphics[width=0.8\linewidth]{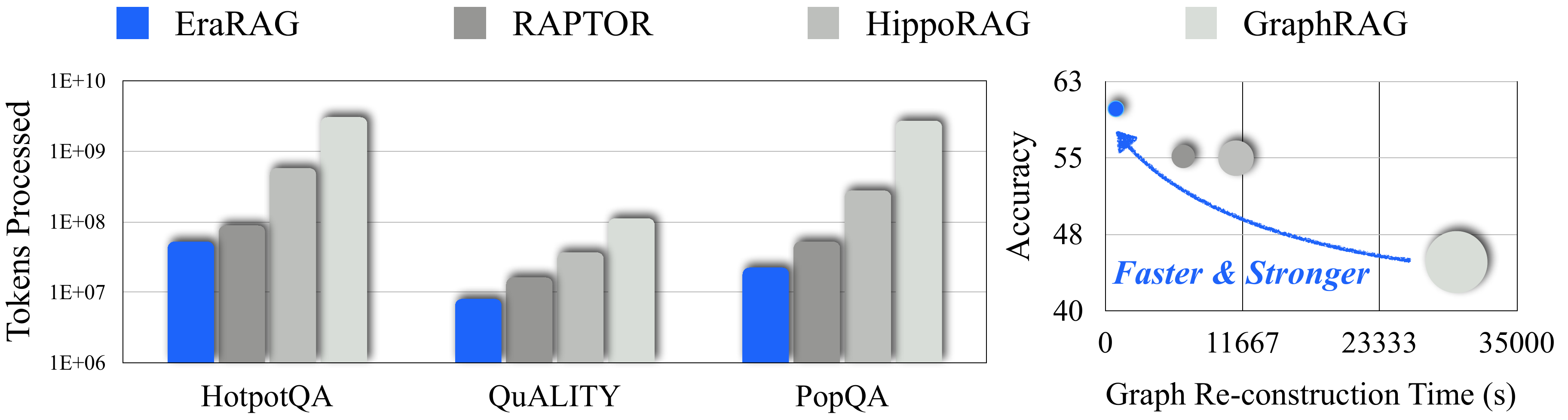}
  \caption{\textbf{Token processed (left)} of EraRAG and baselines via initial graph construction and 10 consecutive insertions. \textbf{Detailed performance (right)} of EraRAG and corresponding baselines on QuALITY. The size of each circle represents the total tokens processed.}
  \label{fig:first}
\end{figure*}

Graph-based RAG methods, despite their promising performance, still face significant challenges in scenarios involving growing corpora. Typical examples include daily additions to news collections, the constant influx of user-generated content on online platforms, and the accumulation of newly published research papers in academic repositories. 
For example, in the {\it Computation and Language} area alone, arXiv typically receives over 100 new paper submissions per day \cite{arxiv}, underscoring the need for efficient graph-based RAG methods capable of adapting to the growing corpora.
%
As illustrated in Figure~\ref{fig:example}, even a minor update to the underlying corpus typically necessitates a complete reconstruction of the graph in existing methods, resulting in substantial computational overhead. Although some prior work has explored dynamic analysis of changing corpora~\cite{dragin}, these approaches still suffer from high costs due to frequent and heavy structural updates.

To address this challenge, we propose \texttt{EraRAG}, a novel multi-layer graph construction framework that integrates hyperplane-based Locality Sensitive Hashing (LSH) for semantic similarity grouping. The overall architecture of \texttt{EraRAG} is illustrated in Figure~\ref{fig:overview}. \texttt{EraRAG} leverages hyperplane-based LSH and controllable partitioning to build a semantically structured graph with consistent granularity. LSH enables efficient chunk grouping, while the size thresholds ensure that each group maintains a consistent level of granularity and semantic abstraction—by containing a similar number of chunks with comparable similarity—we note that our design aligns with prior successes in graph-based RAG area \cite{orig,gr3,grag}, which leverage high-level abstractions to support multi-hop reasoning.
In addition, our LSH-based grouping approach reduces redundancy and improves retrieval accuracy without incurring costly recomputation, by efficiently exploiting the semantic structure of the corpus.

Crucially, \texttt{EraRAG} with hyperplane-based LSH supports fast, localized updates when new corpus entries arrive. Specifically, it encodes the new chunks into vector embeddings, inserts them into the appropriate buckets, and performs upward-propagating adjustments that are confined to the affected segments, without altering unrelated parts of the graph. This localized update strategy significantly improves efficiency by eliminating the need for costly global recomputation. To evaluate the effectiveness of \texttt{EraRAG} in growing-corpus scenarios, we divide the entire corpus into two parts: 50\% is used as the initial corpus, and the remaining 50\% serves as the growing portion. We simulate corpus expansion by incrementally inserting 5\% of the corpus at each step, resulting in ten rounds of insertion. As shown in Figure~\ref{fig:first}, \texttt{EraRAG} consumes far fewer tokens and requires substantially less running time in growing-corpus scenarios compared to existing methods, achieving state-of-the-art accuracy performance on challenging datasets such as QuALITY.
Our main contributions are summarized as follows:
\begin{itemize}[left=0pt]
\item\textbf{LSH-based Graph Construction Framework.} 
    We propose \verb|EraRAG|, a framework that constructs a multi-layered graph through recursive LSH-based segmentation and summarization. This structure not only preserves local and global semantic relationships for accurate retrieval, but also supports efficient, scalable updates when new content is introduced.

\item\textbf{Efficient Incremental Graph Update Mechanism.} 
    \texttt{EraRAG} enables fast and localized updates by combining hyperplane-based LSH with a merge-and-split strategy governed by tunable size thresholds. This design ensures consistent segment granularity, avoids unnecessary recomputation, and supports seamless integration of new corpus entries.

\item\textbf{Extensive evaluation on real-world benchmarks.} Experiments across multiple QA benchmarks demonstrate that \texttt{EraRAG} maintains strong retrieval accuracy in static settings, while in dynamic scenarios, it achieves an order of magnitude reduction in both update time and corresponding token costs compared to other methods, without sacrificing query quality.
\end{itemize}

\section{Related Work}
\noindent\textbf{Graph-based Retrieval Augmented Generation.} When faced with domain-specific or multi-hop queries, large language models often suffer from factual inconsistency or hallucinations—generating confident but inaccurate or nonsensical answers~\cite{hallu,hallu1}. These shortcomings arise from the static nature of LLM pretraining, which limits access to up-to-date or domain-specialized information. To address this issue, \textit{Retrieval-Augmented Generation} (RAG)~\cite{rag1,rag2,rag3} has emerged as a powerful framework that augments LLMs with access to an external knowledge corpus, enabling them to generate more accurate and contextually grounded responses. Typical RAG systems (e.g., Vanilla RAG) consist of the following stages.

\begin{enumerate}[left=0pt]
    \item \emph{Corpus Preprocessing:} The input corpus is first segmented into smaller units known as chunks for better retrieval. Each chunk is then embedded into a dense vector representation using a pre-defined embedding model. These vectors, together with optional metadata, are indexed and stored in a vector database for efficient retrieval. 
    
    \item \emph{Query-time Retrieval:} Upon receiving a user query, the same embedding model is used to encode the query into a vector. This vector is then used to retrieve the top-$k$ most similar chunks from the vector database—typically based on cosine similarity or other distance metrics. The retrieved chunks serve as external knowledge relevant to the query.
    
    \item \emph{Answer Generation:} The original question and the retrieved chunks are formatted into a structured prompt and passed into a language model. The LLM utilizes this information to generate an answer that is ideally more factual, contextualized, and grounded in the retrieved content.
\end{enumerate}

Despite its effectiveness, conventional RAG systems often retrieve semantically redundant or disconnected chunks, limiting their ability to support multi-hop reasoning and coherent generation. To address these limitations, Graph-based RAG~\cite{gr1} is introduced as a structured retrieval paradigm that models semantic relationships through graph-based organization. Contrary to the normal RAG framework, the corpus preprocessing stage of Graph-based RAG transforms the raw corpus into a graph or hierarchical structure, enabling more efficient and accurate retrieval during the generation phase~\cite{gr2}. By encoding semantic relationships between documents, passages, or entities ahead of time, Graph-based RAG reduces redundancy and improves the contextual coherence of retrieved results. This offline organization significantly accelerates retrieval at inference time and enhances the relevance of the supporting evidence, leading to improved response quality.

\noindent\textbf{Locality-Sensitive Hashing.} LSH \cite{IndykM98} is an efficient method for indexing high-dimensional data. The technique leverages hashing to map similar items to the same buckets with high probability. Variants such as E2LSH \cite{datar2004locality} and FALCONN \cite{AndoniILRS15} have gained attention for applications in approximate high-dimensional data retrieval. These methods offer tunable performance and theoretical guarantees but require significant redundancy and additional space cost to ensure accuracy. Unlike the traditional methods of applying LSH to high-dimensional vector retrieval, our method adopts a novel multi-layer framework and dynamic segmentation technology specifically tailored for the RAG system.


\noindent\textbf{Dynamic Retrieval.}
Recent research has focused on dynamic retrieval mechanisms that adapt to the evolving query context or model state during inference, aiming to enhance retrieval relevance and efficiency in context-dependent tasks. DRAGIN~\cite{dragin} detects information needs in real time via attention and uncertainty signals, triggering retrieval only when necessary, and formulates queries dynamically to minimize noise. LightRAG~\cite{lrag}, a graph-based method, introduces a modular retriever design that enables dynamic addition of new documents without rebuilding the full index, making it suitable for evolving corpora. DyPRAG~\cite{dyprag} dynamically injects retrieved content as lightweight parameter adapters into the language model during inference, enabling knowledge integration without altering the core model. However, these approaches largely overlook the consumption of dynamic updates under high-frequency data changes.

\begin{figure*}[t]
  \centering
  \includegraphics[width=\linewidth]{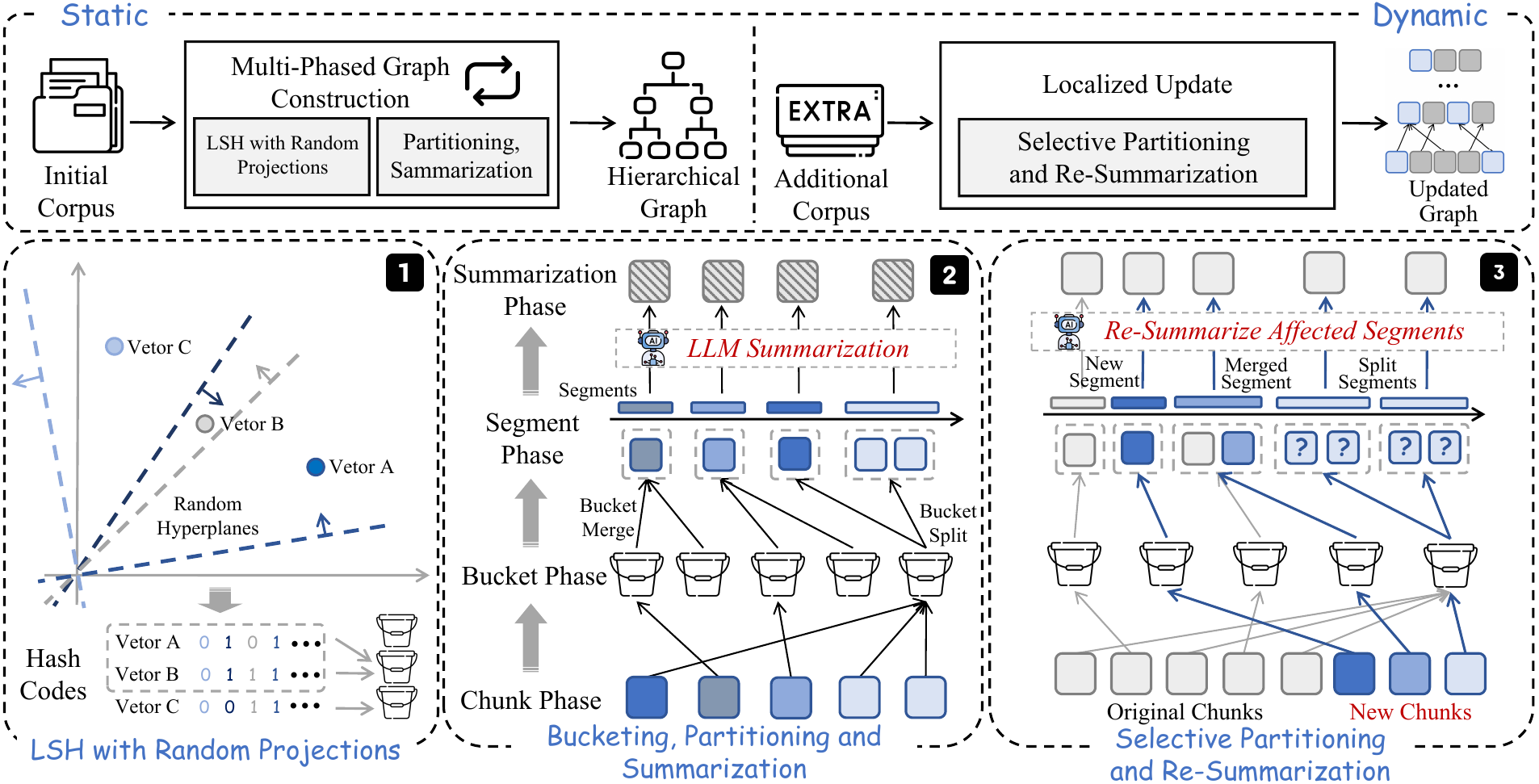}
  \caption{\textbf{Overview of EraRAG}. The framework constructs a hierarchical retrieval graph. In the static mode, initial chunks are bucketed via LSH with random hyperplane projections, and then iteratively partitioned and summarized through controlled bucket splitting and merging. In the dynamic mode, new data can be inserted by selectively re-partitioning and re-summarizing affected segments, enabling efficient updates with minimal overhead.}
  \label{fig:overview}
\end{figure*}

\section{Our Solution}

\subsection{High Level Idea}

The overall architecture of \texttt{EraRAG} is illustrated in Figure~\ref{fig:overview}. Given an input corpus, we first process it into textual chunks, and then encode them into vector embeddings. These embeddings are then processed via a hyperplane-based LSH scheme: each vector is projected onto $n$ randomly sampled hyperplanes and encoded as an $n$-bit binary hash code. Vectors with similar hashes—measured by Hamming distance—are grouped into the same bucket. Since bucket sizes vary depending on semantic similarity within the corpus, a second-stage partitioning is performed to produce the final segments. Segments are constrained by user-defined size bounds: small buckets are merged with adjacent ones, while large buckets are split. For each resulting segment, an LLM is used to summarize its constituent chunks into a new chunk. Built on the recursive construction architecture of RAPTOR~\cite{raptor}, this process of hashing, partitioning, and summarization is recursively applied to construct a multi-layered hierarchical graph, with each layer consisting of a certain granularity of the given corpus.

For dynamic updates, our system reuses the original set of hyperplanes to maintain consistency with the proposed LSH process. Newly added chunks are projected using the same hyperplanes and inserted into the corresponding buckets, which are then re-partitioned as necessary. Segments that are either newly assigned chunks or affected by bucket-level merging or splitting are re-summarized, and their parent nodes are marked as affected. These parent nodes are subsequently recursively re-hashed, re-partitioned, and re-summarized, propagating the changes upward throughout the graph. This approach facilitates localized updates, preserving the structural integrity of the graph while avoiding the need for a full reconstruction.

In the query processing stage, \texttt{EraRAG} adopts a collapsed graph search strategy, in which all nodes are treated uniformly within a flat retrieval space. Upon receiving a query, it is first encoded into an embedding vector using the same encoder employed during graph construction \cite{li2023gslb}. This query embedding is used to retrieve the top-$k$ most similar node embeddings from the vector database under a predefined token budget, selecting the most relevant chunks. The retrieved chunks are concatenated into a single context and passed to the language model together with the original query to generate the final response. This approach enables \texttt{EraRAG} to flexibly accommodate queries of varying granularity by leveraging the multi-level semantics encoded in the graph structure. Furthermore, \texttt{EraRAG} supports an optional biased retrieval strategy that allows users to adjust the proportion of retrieved detailed or summarized chunks based on prior knowledge of the query type.

\begin{table}[t]
\centering
\caption{Frequently Used Notations}
\label{tab:notation}
\scalebox{1.22}{
\begin{tabular}{cl}
\toprule
\textbf{Symbol} & \textbf{Definition} \\
\midrule
$\mathcal{C}$ & Input corpus (collection of text chunks) \\
$c_i$ & $i$-th chunk from corpus \\
$v_i \in \mathbb{R}^d$ & Normalized embedding vector for $c_i$ \\
$h_j \in \mathbb{R}^d$ & $j$-th random hyperplane \\
$k$ & Number of hyperplanes \\
$b_i \in \{0,1\}^k$ & Binary hash code of $v_i$ \\
$\mathcal{B}_{b_i}$ & Bucket indexed by binary code $b_i$ \\
$S_{\min}, S_{\max}$ & Lower and upper bounds for bucket size \\
$\mathcal{S}_i$ & Final adjusted segment (bucket) \\
$s_i$ & Summarized node from segment $\mathcal{S}_i$ \\
$G_\ell$ & Set of graph nodes at layer $\ell$ \\
$L$ & Total number of graph layers \\
$d$ & Embedding dimensionality \\
\bottomrule
\end{tabular}
}
\end{table}

\subsection{Hyperplane-based LSH for Reproducible Grouping}

\noindent The grouping phase plays a pivotal role in the efficient construction of graphs within \texttt{EraRAG}, as it directly influences how semantically similar chunk embeddings are organized for subsequent retrieval tasks. As previously highlighted, Locality-Sensitive Hashing stands out as a highly effective technique for the rapid and high-quality grouping of high-dimensional embeddings, making it a widely adopted approach in large-scale clustering and retrieval systems. Formally, an LSH family is defined as follows \cite{lsh1}:

\begin{definition}[Locality-Sensitive Hashing (LSH) Family]
Let $(\mathcal{X}, D)$ be a metric space and let $r, c > 0$ and $0 < P_1, P_2 < 1$ with $P_1 > P_2$. A family of hash functions $\mathcal{H} = \{ h: \mathcal{X} \rightarrow U \}$ is called $(r, cr, P_1, P_2)$-sensitive if for any $x, y \in \mathcal{X}$:

\begin{itemize}
  \item If $D(x, y) \leq r$, then $\Pr_{h \in \mathcal{H}}[h(x) = h(y)] \geq P_1$,
  \item If $D(x, y) > cr$, then $\Pr_{h \in \mathcal{H}}[h(x) = h(y)] \leq P_2$.
\end{itemize}
\end{definition}

However, conventional LSH methods are not well-suited for clustering and managing text embedding vectors in RAG scenarios. For instance, typical LSH approaches often employ uneven bucket assignment strategies, leading to some buckets containing a large number of elements while others contain few. In RAG settings, where a summary must be generated for each group to ensure high-quality and diverse retrieval corpora, it is crucial that the number of elements per group remains balanced. Moreover, in dynamic text corpora where new documents are continuously added, traditional LSH clustering faces a significant limitation: its lack of reproducibility. Specifically, the non-deterministic nature of bucket assignments means that the addition of new data requires a complete reconstruction of the graph, as existing clusters may be altered unpredictably.


To address these challenges, we propose a hyperplane-based LSH method that provides control over the number of elements in each group and supports efficient updates. The methodology for our grouping technique is illustrated in Section 1 of Figure~\ref{fig:overview}, where we project the high-dimensional embeddings of chunks onto a set of randomly sampled hyperplanes. This projection produces compact binary hash codes that facilitate the rapid organization of embeddings into consistent clusters. Each chunk embedding \( v_i \in \mathbb{R}^d \) is mapped to a $k$-bit code through the following procedure:
\[
\text{hash}(v) = [\text{sign}(v \cdot h_1), \cdots, \text{sign}(v \cdot h_k)]
\]
where \( \{h_1, \ldots, h_k\} \) are hyperplanes randomly drawn from \( \mathbb{R}^d \). Each bit in the generated hash corresponds to the sign of the dot product between the embedding and a hyperplane, determining on which side of the hyperplane the embedding lies. These binary codes function as bucket identifiers, effectively grouping semantically similar embeddings within the Hamming space. This method ensures the preservation of angular proximity: embeddings with smaller angular distances—i.e., higher cosine similarity—tend to produce hash codes that differ by fewer bits. More formally, for two normalized vectors \( v_1 \) and \( v_2 \), the probability that they are assigned the same bit on a randomly selected hyperplane is given by the following \cite{lsh1}:

\begin{theorem}
    Given two normalized vectors \( v_1, v_2 \in \mathbb{R}^d \) and a random hyperplane \( h \), the probability that both vectors lie on the same side of \( h \) is:
    \[
    P(h(v_1) = h(v_2)) = \frac{1 + \cos(\theta)}{2},
    \]
    where \( \theta \) represents the angle between \( v_1 \) and \( v_2 \).
\end{theorem}

This characteristic guarantees that vectors with greater similarity are more likely to be assigned to the same bucket. Crucially, unlike conventional LSH implementations that discard projection information after the hashing step, our approach preserves the random hyperplanes used during hashing. This design ensures full reproducibility of the clustering process, allowing new embeddings to be consistently and deterministically assigned to the correct buckets without recomputing the entire corpus. By preserving the hyperplanes, our method enables efficient incremental updates to the graph, supporting dynamic changes in evolving corpora without requiring full reconstruction. Such reproducibility is crucial for maintaining the integrity of the grouping process during updates.


\subsection{Bucket Partitioning and Multilayer Graph Construction}

\begin{algorithm}[t]
\caption{Hyperplane-based LSH Segmentation}
\label{algo:lsh-tree}
\KwIn{Corpus $C$, number of hyperplanes $n$, size bounds $[S_{\min}, S_{\max}]$, max depth $L$}
\KwOut{Hierarchical LSH Graph $G$ with $L$ layers}

Tokenize $C$ into text chunks $\{c_i\}$\;
Compute normalized embeddings $\{v_i\}$ for all chunks\;

Sample $n$ random hyperplanes $\{h_j\}_{j=1}^{n}$\;
\For{each vector $v_i$}{
  Project $v_i$ onto hyperplanes to obtain hash code $b_i$ \tcp*[r]{via $\text{sign}(v_i \cdot h_j)$}
  Assign $v_i$ to bucket $B_{b_i}$ based on hash
}

\For{each bucket $B$}{
  \uIf{$|B| > S_{\max}$}{
    Split $B$ into smaller buckets of size $\leq S_{\max}$
  }
  \uElseIf{$|B| < S_{\min}$}{
    Merge $B$ with adjacent buckets until $\geq S_{\min}$
  }
}

\For{each adjusted bucket (segment) $S$}{
  Summarize chunks in $S$ using LLM $\to$ summary chunk $s_S$\;
}

Set $G_0 = \{s_S\}$ \tcp*[r]{This forms the layer-0 leaf nodes}
\For{$l = 1$ to $L$}{
  \uIf{stopping criterion met ($|G_{l-1}| < d + 1$)}{
    \Return final graph $G$
  }
  Compute embeddings for all chunks in $G_{l-1}$\;
  Repeat hashing, partitioning and summarizing (rows 4-13) to obtain new summarized nodes $G_l$\;
}
\Return Finalized graph {$G = \{G_0, G_1, ..., G_L\}$}
\end{algorithm}

\noindent Based on the proposed LSH-based grouping mechanism, the initial graph construction process can be outlined in Algorithm~\ref{algo:lsh-tree}. Following the initial grouping of chunk embeddings into buckets, we perform a secondary partitioning step to transform these raw buckets into well-structured segments suitable for hierarchical graph construction (Lines 7-11). Departing from conventional LSH-based clustering, our approach introduces an additional partitioning mechanism to regulate both the size and semantic consistency of each segment. This is essential because the number of chunks in each bucket can vary significantly due to uneven semantic density across the corpus.

Formally, let \( \mathcal{B} = \{B_1, B_2, \dots, B_m\} \) be the set of initial buckets derived from LSH hashing. For each bucket \( B_i \), we introduce user-defined lower and upper bounds \( S_{\min} \) and \( S_{\max} \) on acceptable segment sizes, where both $S_{min}$ and $S_{max}$ are $\Theta(c)$, and $c$ is a user-defined parameter. If \( |B_i| < S_{\min} \), the bucket is merged with adjacent ones \( B_{i-1} \) or \( B_{i+1} \) based on proximity in Hamming space. Conversely, if \( |B_i| > S_{\max} \), we split it into sub-buckets \( B_i^{(1)} \), \( B_i^{(2)} \). This yields a final set of segments \( \mathcal{S} = \{S_1, S_2, \dots, S_n\} \), each containing a manageable and semantically consistent group of chunks.

The choice of segment size bounds \( t_{\min} \) and \( t_{\max} \) critically influences the resulting graph structure. Narrow bounds enforce uniform segment sizes, yielding a well-balanced hierarchy with consistent abstraction across layers. However, this strictness often necessitates excessive merging and splitting, potentially grouping semantically dissimilar chunks and degrading summarization quality. In contrast, wider bounds preserve intra-segment coherence and improve summarization fidelity, but may result in structurally imbalanced graphs, where uneven segment sizes lead to inconsistent abstraction and suboptimal retrieval performance. This trade-off highlights the tension between structural regularity and semantic coherence, both of which are critical to effective hierarchical representation and retrieval. This will be further studied in the experiment section.

After segmentation, each segment \( S_i \) is summarized into a new chunk \( c_i^{(1)} \) via a large language model (Line 12-14). The embedding of the summarized chunk, \( v_i^{(1)} = \text{encode}(c_i^{(1)}) \), is then re-hashed using the same set of LSH hyperplanes. The entire process of hashing, partitioning, and summarizing is applied recursively to construct a multi-layered graph structure \( \mathcal{G} \), where each successive layer encodes progressively coarser semantic abstractions of the corpus (Lines 15-19). This recursive summarization process results in a hierarchical graph capable of handling both detailed and high-level queries.

\begin{theorem}
Let $|C|$ denote the number of text chunks in corpus $C$, $d$ be the embedding dimension, $n$ be the number of hyperplanes, $L$ be the user-defined maximum depth, and $\mathcal{S}_{\mathrm{LLM}}$ the amortised time required by the LLM to summarise \emph{one} segment.  The time complexity of Algorithm~\ref{algo:lsh-tree} is $O\!\bigl(|C|\,(n\,d+\mathcal{S}_{\mathrm{LLM}})\bigr)$ and the space complexity is $O(|C|\,d)$.
\end{theorem}

\begin{proof} Let $N_\ell$ denote the number of chunks present at level $\ell$, with $N_0=|C|$. Processing a single level consists of three dominant actions. First, every chunk is embedded (or its cached embedding is reused), costing $O(N_\ell d)$. Second, each embedding is projected onto the $n$ hyper-planes to form its binary hash, adding another $O(N_\ell n d)$ operations. Third, all resulting segments are summarised once by the LLM; the number of freshly created parent nodes is $N_{\ell+1}$, so this step takes $O(N_{\ell+1}\mathcal{S}_{\mathrm{LLM}})$ time. Thus the total work at level $\ell$ is $O\!\bigl(N_\ell n d+N_{\ell+1}\mathcal{S}_{\mathrm{LLM}}\bigr)$. Because every summarised node must aggregate at least $S_{\min}>1$ children, we have the geometric decay $N_{\ell+1}\le N_\ell/S_{\min}$. Substituting this inequality and summing over all levels yields 
\begin{align}
T(|C|) &\le \sum_{\ell\ge0} \Bigl(N_\ell n d+N_{\ell+1}\mathcal{S}_{\mathrm{LLM}}\Bigr) \notag \\
      &\le |C| \Bigl(n d+\mathcal{S}_{\mathrm{LLM}}/S_{\min}\Bigr) \sum_{\ell\ge0} S_{\min}^{-\ell} \notag \\
      &= O\bigl(|C|\,(n d+\mathcal{S}_{\mathrm{LLM}})\bigr), \notag
\end{align}
because the geometric series $\sum_{\ell\ge0}S_{\min}^{-\ell}=1/(1-1/S_{\min})$ is a constant independent of $|C|$, $n$, or $d$. 

For space, the algorithm stores one $d$-dimensional vector per live chunk plus the $n$ hyper-plane normals. The largest number of simultaneously live chunks occurs at the input layer and equals $|C|$, so the peak memory footprint is $|C|d+nd=O(|C| d)$, completing the proof. 
\end{proof}

\subsection{Query processing for \texttt{EraRAG}}










\begin{algorithm}[t]
\caption{Query Processing for \texttt{EraRAG}}
\label{algo:query}

\KwIn{Query $q$, vector database $\mathcal{V}$, retrieval size $k$, token budget $T$}
\KwOut{Final answer $a_q$ generated by LLM}

Encode the query: $\mathbf{e}_q \leftarrow \texttt{encode}(q)$\;

Retrieve top-$k$ candidates from $\mathcal{V}$ under token budget $T$: $\mathcal{R}_q \leftarrow \texttt{vectordb\_search}(\mathbf{e}_q, k, T)$\;

Concatenate retrieved chunks: $\mathcal{C}_q \leftarrow \texttt{concat}(\mathcal{R}_q)$\;

Generate answer using LLM: $a_q \leftarrow \mathcal{M}(q, \mathcal{C}_q)$\;

\Return{$a_q$}\;
\end{algorithm}

In the retrieval stage, various methods have been proposed for navigating recursive hierarchical graphs. Notably, recent work~\cite{raptor} demonstrates that for such graph structures—including the one used in \texttt{EraRAG}—a global collapsed graph search (i.e., flat top-$k$ search) consistently outperforms hierarchical top-down structural search across different chunk sizes. To strike a balance between preserving fine-grained details and capturing high-level semantics, \texttt{EraRAG} adopts the collapsed graph search approach.

The process of query processing for \texttt{EraRAG} is outlined in Algorithm~\ref{algo:query}. Upon receiving a query $q$, \texttt{EraRAG} first encodes it into an embedding vector $\mathbf{e}_q \in \mathbb{R}^d$ using the same encoder employed during graph construction. This embedding is then submitted to a FAISS-based vector database $\mathcal{V}$, which indexes the embeddings $\{\mathbf{e}_i\}_{i=1}^N$ corresponding to all nodes in the collapsed retrieval graph, including both leaf chunks and summary nodes.

Similarity is measured using inner product or cosine similarity, depending on the FAISS index configuration. A top-$k$ retrieval is performed to efficiently select the $k$ most relevant nodes under a predefined token budget $T$. The retrieved chunks are concatenated into a single context, which is passed to the LLM alongside the original query. This collapsed retrieval strategy enables the model to jointly reason over both fine-grained content and high-level semantic abstractions, allowing \texttt{EraRAG} to effectively address diverse query types ranging from detail-oriented factual questions to paragraph-level summarization and reasoning tasks.

Upon further analysis, we observe that for fine-grained queries requiring specific textual details, retrieving from leaf nodes significantly improves the LLM’s ability to generate accurate responses. This can be attributed to the nature of the summarization process, in which certain low-level textual details may be omitted due to information compression. As a result, key information necessary for answering detailed queries may be lost in higher-level summary nodes.

Motivated by this observation, we propose an adaptive retrieval strategy that tailors the retrieval pattern according to the expected granularity of the query. Specifically, we introduce two distinct search patterns for \texttt{EraRAG}, each designed to emphasize different semantic levels of the graph. To support this, we define an additional parameter $p \in [0, 1]$, which controls the proportion of chunks retrieved from different layers, in addition to the top-$k$ retrieval budget.

\begin{itemize}[left=0pt]
    \item \textbf{Detailed search.} For queries that demand fine-grained, factual information, we prioritize retrieving chunks from the leaf layer. A top-$pk$ search is first performed over the leaf layer. The remaining $(1-p)k$ chunks are then selected via top-$(k - pk)$ search over the summarized layers, ensuring that sufficient contextual abstraction is still preserved.
    
    \item \textbf{Summarized search.} For queries that require understanding of high-level semantics or abstract narrative structure, we reverse the retrieval focus. A top-$pk$ search is conducted over the summary layers, followed by a top-$(k - pk)$ retrieval over the leaf nodes to supplement the results with essential factual grounding.
\end{itemize}

This mechanism ensures that a total of $k$ chunks are retrieved per query, while allowing the user to control the trade-off between detailed and generalized information according to the query’s nature. Note that in our experiments, we continue to employ the standard collapsed graph search to maintain general applicability. A more in-depth evaluation of these two adaptive search strategies is provided in the \href{https://github.com/EverM0re/EraRAG-Official}{technical report}.

\begin{theorem}
    
For a collapsed retrieval graph that stores $N$ embedded nodes of dimension $d$, a query requesting the top-$k$ neighbours under token budget $T$ runs in $
T_{\mathrm{query}}
     =O\!\bigl(d\;+\;\mathcal{V}_{\mathrm{search}}(N,d,k)\;+\;
                     \mathcal{S}_{\mathrm{LLM}}(T)\bigr),
$ where $\mathcal{V}_{\mathrm{search}}(N,d,k)$ denotes the time
complexity of the underlying vector database top-$k$ search and
$\mathcal{S}_{\mathrm{LLM}}(T)$ is the latency of the answer-generation
LLM when constrained to at most $T$ output tokens.
\end{theorem}

\begin{proof}
    
The query pipeline starts with an embedding step: the raw text $q$ is fed through the same encoder used during graph construction, producing a $d$-dimensional vector $\mathbf{e}_q$.  This is a single forward pass whose cost scales linearly with the dimension, hence $\Theta(d)$.  The result is immediately normalised (if required by the index) and handed over to the vector database; this normalisation is a constant-factor operation and does not alter the asymptotic bound. 

The core of the procedure is the \texttt{vectordb\_search} invocation.  All distance computations, inverted-list probes, graph traversals, and heap updates incurred while extracting the $k$ nearest neighbours are captured by the term $\mathcal{V}_{\mathrm{search}}(N,d,k)$.  For a brute-force (IndexFlat) configuration this term equals $O(Nd)$, whereas for more sophisticated indices such as IVF-PQ or HNSW it becomes sub-linear in $N$ but still at least linear in $d$ and nearly linear in $k$.  Crucially, the presence of multiple hierarchical layers in \texttt{EraRAG} does not affect this complexity because the collapsed graph is treated as a single flat index of size $N$. 

Once the $k$ most similar nodes are returned, the algorithm merely concatenates their associated texts—an $O(k)$ operation that is dominated by the previous step—and forwards the resulting context, together with the original query, to the language model $\mathcal{M}$.  The generation stage produces at most $T$ tokens and therefore costs $\mathcal{S}_{\mathrm{LLM}}(T)$.  Any overhead introduced by the adaptive detailed/summarised retrieval policy is bounded by an extra scan over the same $k$ results and thus remains $O(k)$.  Summing the costs of encoding, vector search, and response generation yields the claimed overall time complexity: $T_{\mathrm{query}} =O\!\bigl(d+\mathcal{V}_{\mathrm{search}}(N,d,k)+ \mathcal{S}_{\mathrm{LLM}}(T)\bigr)$.
\end{proof}

\subsection{Selective Re-Segmenting and Summarization for Dynamic Corpora}

\begin{algorithm}[t]
\caption{Selective Re-Segmenting and Summarization for Dynamic Corpora}
\label{algo:dynamic-update}

\KwIn{Incremented chunks $c_{\text{new}}$, stored hyperplanes $\{h_j\}_{j=1}^k$, current graph $G$}
\KwOut{Updated graph $G$}
Compute the embedding for the new chunk $v_{\text{new}} = \text{encode}(c_{\text{new}})$ and its hash code $\text{hash}(v_{\text{new}})$\;
Assign $c_{\text{new}}$ to the corresponding leaf bucket based on $\text{hash}(v_{\text{new}})$, mark the incremented buckets as \underline{affected}\;

\For{each \underline{affected} bucket $\mathcal{B}_{b}$}{
    \If{$|\mathcal{B}_b| > S_{\max}$}{
        Split $\mathcal{B}_b$ into smaller buckets of size $\leq S_{\max}$, mark resulting buckets as \underline{affected}\;
    }
    \ElseIf{$|\mathcal{B}_b| < S_{\min}$}{
        Merge $\mathcal{B}_b$ with adjacent buckets until $|\mathcal{B}_b| \geq S_{\min}$, mark $\mathcal{B}_b$ as \underline{affected}\;
    }
    Finalize buckets into segments $\mathcal{S}_i$, the segment is marked as \underline{affected} if concluding \underline{affected} buckets\;
}
\For{$l = 1$ to $L$}{
\For{each \underline{affected} segment $\mathcal{S}_i$}{
    Compute a resummarization of segment $\mathcal{S}_i$:
    \[
    s_i = f_{\text{summarize}}(\text{Children}(\mathcal{S}_i))
    \]
    Delete the original chunk node and add all its children to the new summarized chunk. Mark the new summarized chunk as \underline{affected}\;
    \For{each \underline{affected} chunk $c_{affected}$}{
        Compute its embedding and hash code
        Repeat selective bucketing, segmenting and resummarizing.
    }
}
}
\textbf{Note:} If $l = current max layer$, $l < L$ and $N_{currentlayer} > S_{max}$, create a new layer and conduct another round of summairzation\;
\Return{Updated graph $G$}\;
\end{algorithm}

Graph construction is one of the most time-consuming operations in the Graph-RAG process. However, existing methods fail to address the challenges posed by dynamic corpora. In such cases, even minor additions to the corpus often necessitate a complete reconstruction of the graph, leading to significant time overhead and increased token consumption. To overcome these limitations and facilitate efficient updates in the presence of evolving corpora, \texttt{EraRAG} introduces a selective re-segmenting and re-summarizing mechanism that confines structural modifications to localized regions of the graph. This approach avoids the need for a full graph reconstruction by reusing the hyperplanes \( \{h_1, \dots, h_k\} \) generated during the initial LSH process. By doing so, we ensure the consistent hashing of new chunk embeddings, thus preserving the integrity of the graph while efficiently incorporating new data. The detailed procedure is outlined in Algorithm~\ref{algo:dynamic-update}, which corresponds to Section 3 of Figure~\ref{fig:overview}.

Given a newly added chunk \( c_{\text{new}} \), we compute its embedding \( v_{\text{new}} = \text{encode}(c_{\text{new}}) \), and derive its hash code $\text{hash}(v_{\text{new}})$ via the same LSH process with the hyperplane parameters (Line 1). The new chunk is inserted into the corresponding bucket (or a new bucket) and these buckets are marked as affected (Line 2). The affected buckets are then subjected to the same partitioning logic as in the static phase (Lines 3-8). If a segment is modified due to chunk insertion, merging, or splitting, it is marked as affected and re-summarized using the LLM (Lines 10-11). This localized update propagates hierarchically: when a segment \( S_i \) is updated, its summarized representation \( c_i^{(l)} \) at level \( l \) becomes outdated. For layers above the leaf level, operations other than simple addition are challenging to perform without compromising the integrity of the graph structure. To address this issue, we propose the following solution: when a re-summarization is required, a new node containing the updated summary is created. The original node, which holds the outdated summary, is removed, and all of its child nodes are reassigned to the child list of the new node. This new node is then treated as an incrementally added chunk in the next layer and subsequently undergoes encoding, hashing, bucketing, and partitioning procedures. Letting \( \mathcal{A}(S_i) \) denote the set of ancestors of \( S_i \), we recursively apply the update operation:
\[
\forall S_j \in \mathcal{A}(S_i), \quad \text{ReSummarize}(S_j) \leftarrow f_{\text{summarize}}(\text{Children}(S_j))
\]
In this way, only subgraphs affected by the incremented data are modified, which ensures that updates remain computationally bounded and structurally contained.

This selective propagation enables fast and consistent integration of new corpora, preserving the integrity of unaffected graph regions. As a result, the system maintains both the retrieval quality of its hierarchical representations and the efficiency of graph maintenance.
\begin{theorem}
Let $|C|$ be the number of chunks already stored in the graph, $\Delta$ be the number of newly–arriving chunks handled by a single update call, $d$ be the embedding dimension, $n$ be the number of stored hyper-planes, and $\mathcal{S}_{\mathrm{LLM}}$ the cost of one LLM summarisation. Assuming the size bounds satisfy $1<S_{\min}\le S_{\max}=O(1)$, the time cost of update algorithm is $T_{\text{update}}(\Delta)=O\!\bigl(\Delta\,(n\,d+\,\mathcal{S}_{\mathrm{LLM}})\bigr)$.

\end{theorem}

\begin{proof}
Consider first the operations triggered by a \emph{single} incoming chunk. The algorithm encodes the text into a $d$-dimensional vector and projects it onto $n$ stored hyper-planes, giving the hash code that determines the target bucket. Both the forward pass through the encoder and the $n$ inner products take $O(n d)$ time. Inserting the chunk into the bucket only updates a constant-size header and is
therefore $O(1)$.

The insertion may violate the size bounds of the bucket or of one of its ancestor segments. Because every bucket is limited to $S_{\max}=O(1)$ elements and each segment must contain at least $S_{\min}>1$ children, a split or merge can touch at most a constant number of adjacent buckets, and this perturbation propagates \emph{upwards} through at most the $L$ existing layers. At each layer, in the case of amortization, no more than segments of a constant number become inconsistent. This is because a segment that has just split or merged can accommodate updates and changes after $\Theta(S_{min})$. Therefore, the algorithm performs a re-summary call to the LLM in $O(\mathcal{S}_{{LLM}})$ time for each affected bucket. The purely algorithmic bookkeeping at that layer (creating or deleting a node, updating pointers, and rehashing the new summary) is again $O(1)$. Because the layer depth $L$ is a constant factor, the per-layer cost
is dominated by the LLM call, and the overall per-chunk cost is $O(n d + \,\mathcal{S}_{\mathrm{LLM}})$.

Finally, an update call handles $\Delta$ new chunks independently: no operation performed for one chunk alters the asymptotic amount of work required for another.  Summing the per-chunk cost over the $\Delta$
insertions therefore multiplies it by $\Delta$, which yields the total running time
$T_{\mathrm{update}}(\Delta)=O\!\bigl(\Delta\,(n d+\,\mathcal{S}_{\mathrm{LLM}})\bigr)$.
\end{proof}

\section{Experimental Setup}

\noindent\begin{tikzpicture}
\filldraw (0,0) -- (-0.15,0.08) -- (-0.15,-0.08) -- cycle ; 
\end{tikzpicture} \textbf{Dataset.} We evaluate \texttt{EraRAG}'s performance across the following five real-world question-answering datasets: \ifarxiv
\else
\textbf{PopQA}~\cite{pop} is a 14k-scale open-domain dataset with entity-centric questions from Wikidata, targeting factual recall. \textbf{MultiHopQA}~\cite{multihop} contains 2,556 complex questions requiring information integration across multiple documents. \textbf{HotpotQA}~\cite{hotpot} is a 113k-scale Wikipedia-based dataset focused on multi-hop reasoning and compositional QA. \textbf{QuALITY}~\cite{pang2021quality} is a multiple-choice dataset based on long-form documents, designed to assess deep reading comprehension. \textbf{MuSiQue}~\cite{musique} includes 25,000 multi-hop questions requiring reasoning over multiple facts, emphasizing logical composition and connection. 
\fi

\noindent\begin{tikzpicture}
\filldraw (0,0) -- (-0.15,0.08) -- (-0.15,-0.08) -- cycle ; 
\end{tikzpicture} \textbf{Baseline.} We evaluate \texttt{EraRAG} against a range of baselines grouped into three categories. \emph{Inference-only methods} include ZeroShot and Chain-of-Thought (CoT)~\cite{cot}, which rely solely on the language model's reasoning without external retrieval. \emph{Retrieval-only methods} such as BM25~\cite{bm25} and Vanilla RAG~\cite{vanilla} enhance the input using sparse or dense retrieval, but do not incorporate structural reasoning. \emph{Graph-based RAG methods} include GraphRAG~\cite{grag}, HippoRAG~\cite{hipporag}, RAPTOR~\cite{raptor}, and LightRAG~\cite{lrag}, which use graph structures to improve retrieval quality and multi-hop reasoning. Variants of LightRAG (i.e., Local, Global, and Hybrid) are  denoted as LightRAG-L/G/H for short.

\noindent\begin{tikzpicture}
\filldraw (0,0) -- (-0.15,0.08) -- (-0.15,-0.08) -- cycle ; 
\end{tikzpicture} \textbf{Metric.}
Following the evaluation protocols in~\cite{schick2023toolformer, mallen2022not}, we use Accuracy and Recall as performance metrics for the selected question answering datasets. Instead of requiring exact string matches, a prediction is considered correct if it contains the gold answer, enabling a more flexible assessment of answer relevance. Note that Recall is not reported for the QuALITY dataset, as it does not provide an exhaustive set of valid reference answers, making it infeasible to determine the proportion of relevant information retrieved. Accordingly, only Accuracy is used for this dataset.

\noindent\begin{tikzpicture}
\filldraw (0,0) -- (-0.15,0.08) -- (-0.15,-0.08) -- cycle ; 
\end{tikzpicture} \textbf{Implementation Details.} \ifarxiv
\else
Llama-3.1-8B-Instruct-Turbo~\cite{kassianik2025llama} is used as the default LLM for all experiments, as it is widely adopted in recent RAG research~\cite{marom2025general}. For text representation, we employ BGE-M3~\cite{multi2024m3}, a state-of-the-art embedding model that supports both multilingual and multi-granularity retrieval. To ensure fair comparison and consistent evaluation across RAG baselines, all methods are implemented within the unified framework proposed in~\cite{orig}, which provides a systematic platform for integrating and benchmarking both graph-based and non-graph-based retrieval-augmented generation architectures. We define token consumption as the sum of the input prompt tokens and the output tokens, while the graph building time refers to the time elapsed from chunking to the completion of the graph construction.
\fi

\begin{table*}[t]
  \centering
  \caption{QA performance (Accuracy and Recall) on RAG benchmarks using Llama-3.1-8B-Instruct-Turbo as the QA reader. The \textbf{best} and \underline{second-best} results are highlighted.}
  \label{tab:qa-performance}
  \renewcommand{\arraystretch}{1.15} 
  \scalebox{1.15}{
  \begin{tabular}{cc|cc|c|cc|cc|cc}
    \toprule
    \multicolumn{2}{c|}{Baseline} 
    & \multicolumn{2}{c|}{PopQA} 
    & QuALITY 
    & \multicolumn{2}{c|}{HotpotQA} 
    & \multicolumn{2}{c|}{MuSiQue} 
    & \multicolumn{2}{c}{MultihopQA} \\
    
    Type & Method & Acc & Rec & Acc & Acc & Rec & Acc & Rec & Acc & Rec \\
    \midrule

    \multirow{2}{*}{Inference-only} 
    & ZeroShot & 29.39 & 9.73 & 38.21 & 32.03 & 43.57 & 3.92 & 6.90 & 45.29 & 25.02 \\
    & CoT & 50.23 & 19.85 & 41.88 & 34.90 & 41.25 & 9.25 & 28.30 & 49.32 & 29.88 \\
    
    \midrule
    \multirow{2}{*}{Retrieval-only} 
    & BM25 & 45.09 & 21.20 & 40.24 & 36.30 & 50.09 & 13.68 & 10.58 & 42.90 & 18.37 \\
    & Vanilla RAG & 56.21 & 27.85 & 39.87 & 48.32 & 55.28 & 14.22 & 29.58 & 50.50 & 37.87 \\

    \midrule
    \multirow{6}{*}{Graph-based} 
    & LightRAG-L & 38.92 & 11.55 & 32.73 & 30.26 & 35.77 & 9.21 & 15.32 & 44.05 & 30.70 \\
    & LightRAG-G & 33.29 & 15.21 & 34.30 & 28.33 & 41.52 & 7.89 & 20.97 & 42.48 & 37.21 \\
    & LightRAG-H & 37.02 & 17.35 & 33.22 & 32.02 & 39.90 & 10.24 & 19.80 & 45.20 & 39.81 \\
    & GraphRAG & 49.98 & 21.28 & 44.90 & 40.84 & 47.39 & 19.32 & 28.81 & 56.98 & \textbf{45.53} \\
    & HippoRAG & \underline{59.29} & 25.88 & 53.31 & 50.46 & 56.12 & \underline{25.15} & \underline{39.71} & 57.49 & 42.17 \\
    & RAPTOR & 59.02 & \underline{27.34} & \underline{55.48} & \underline{53.29} & \textbf{61.97} & 24.02 & 37.92 & \underline{60.11} & 40.82 \\

    \midrule
    \skyblue{}\multirow{1}{*}{Our proposed} 
    & \textbf{EraRAG} & \textbf{62.98} & \textbf{28.54} & \textbf{60.25} & \textbf{55.39} & \underline{61.43} & \textbf{25.39} & \textbf{41.35} & \textbf{62.87} & \underline{42.98} \\
    
    \bottomrule
  \end{tabular}
  }
\end{table*}

\section{Experimental Results}

We now present the results of static QA evaluation and dynamic insertion experiments, through which we assess the update efficiency and dynamic structural robustness of EraRAG.

\noindent\begin{tikzpicture}
\filldraw (0,0) -- (-0.15,0.08) -- (-0.15,-0.08) -- cycle ; 
\end{tikzpicture} \textbf{Static QA Performance.} \ifarxiv
\else
Table~\ref{tab:qa-performance} summarizes the QA results of EraRAG and baselines on five benchmarks. EraRAG consistently outperforms all methods in most cases, showing significant gains in Accuracy and Recall. Inference-only models perform poorly on open-domain and multi-hop tasks due to lacking external retrieval, while retrieval-only methods offer moderate improvements but are limited by weak structural reasoning. Graph-based RAGs achieve stronger results overall. Yet, \texttt{EraRAG} surpasses all baselines on 8 of 10 metrics, notably improving QuALITY Accuracy by 4.8\% over RAPTOR. This is largely due to its segmentation strategy: unlike RAPTOR's overlapping clustering, which increases coverage but introduces redundancy, EraRAG uses one-to-one assignments with size constraints for coherent segmentation and stable hierarchy. This controlled granularity enhances summarization and retrieval, especially on complex datasets.

Given that the majority of baselines exhibit limited performance and do not approach competitive levels, subsequent dynamic evaluation will be restricted to the strongest graph-based methods: GraphRAG, HippoRAG, and RAPTOR.
\fi

\begin{figure}[t]
  \centering
  \vspace{0pt}
  \includegraphics[width=0.45\textwidth]{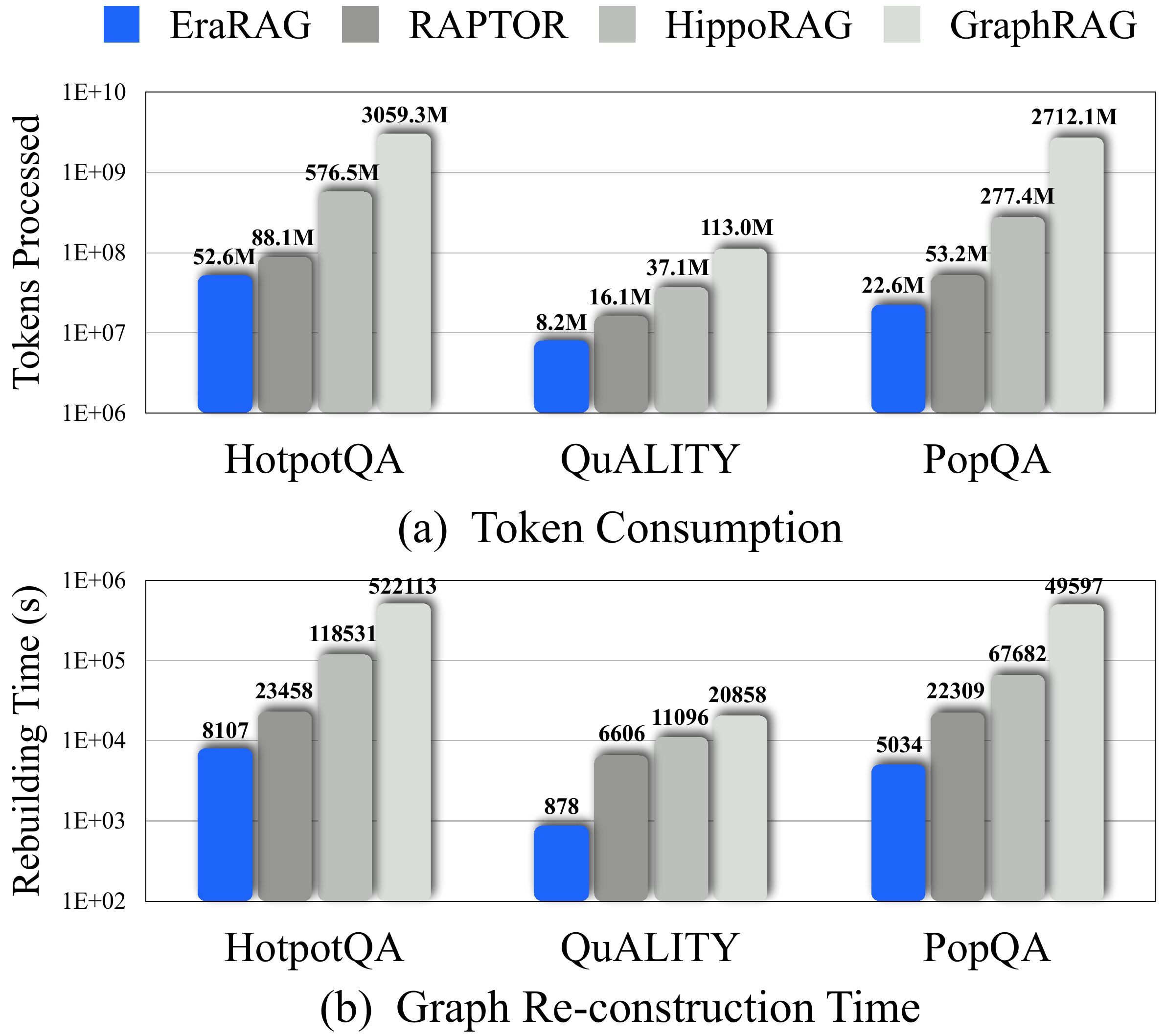}
  \caption{Token cost and graph rebuild time throughout insertions.}
  \label{fig:cost}
  \vspace{0pt}
\end{figure}

\noindent\begin{tikzpicture}
\filldraw (0,0) -- (-0.15,0.08) -- (-0.15,-0.08) -- cycle ; 
\end{tikzpicture} \textbf{Dynamic Insertion Consumption.} To evaluate the dynamic update efficiency of EraRAG, we design an experiment simulating real-world scenarios where new corpora are incrementally added. Each QA dataset is split into two halves: the first 50\% is used to construct the initial graph, and the remaining 50\% is divided into ten equal segments representing sequential updates. For baselines without dynamic support, each update reconstructs the graph from scratch, including the base 50\% and an additional 5\%, simulating cumulative growth. Evaluations are conducted on HotpotQA, PopQA, and QuALITY, recording token consumption and graph construction time at each stage. Note that only graph construction time is measured, excluding preprocessing steps like community detection in GraphRAG.

Figure~\ref{fig:cost} shows that EraRAG consistently achieves the lowest token and time cost across datasets. Compared to RAPTOR, EraRAG reduces token usage by up to 57.6\% (on PopQA) and graph rebuilding time by 77.5\% (on QuALITY). While RAPTOR is already lightweight, EraRAG’s selective reconstruction further improves efficiency. In contrast, GraphRAG and HippoRAG incur significantly higher costs: GraphRAG performs full re-clustering after each update, causing excessive time and memory usage; HippoRAG, though incremental, involves repeated path expansion and semantic filtering, leading to inflated token usage. These results demonstrate that EraRAG's selective update mechanism is well-suited for efficient adaptation to evolving corpora in practical deployments.

\noindent\begin{tikzpicture}
\filldraw (0,0) -- (-0.15,0.08) -- (-0.15,-0.08) -- cycle ; 
\end{tikzpicture} \textbf{Incremental Performance Evaluation.}

Having established the efficiency of our approach, we now turn to evaluating the effectiveness of EraRAG's selective re-construction mechanism. Specifically, we examine its ability to incorporate new information into the existing graph without disrupting previously established structures. To this end, we perform an incremental performance evaluation. In contrast to the preceding experiment, which focuses on computational efficiency and construction cost, this evaluation assesses retrieval quality by measuring Accuracy and Recall after the initial graph construction and following each dynamic update. 

Experimental results on HotpotQA, PopQA, and QuALITY are presented in Figure~\ref{fig:perf}. In each subplot, the dotted horizontal lines represent the Accuracy and Recall obtained from a full static graph built using the complete corpus in one go. The solid lines indicate the incremental performance of \texttt{EraRAG} as new data segments are dynamically inserted in stages.

Across all three datasets, both Accuracy and Recall curves show a clear upward trend, indicating that each incremental addition contributes additional useful information to the graph and progressively improves retrieval quality. This confirms that the selective re-construction mechanism effectively incorporates new content without degrading existing structures. Also note that the final retrieval performance after the last update stage nearly converges to the corresponding static upper bound. That is, \texttt{EraRAG} not only maintains structural integrity throughout the update process but also achieves retrieval effectiveness comparable to that of a fully reconstructed graph. These results highlight the robustness of our selective updating strategy.

\begin{figure*}
  \centering
  \includegraphics[width=0.75\linewidth]{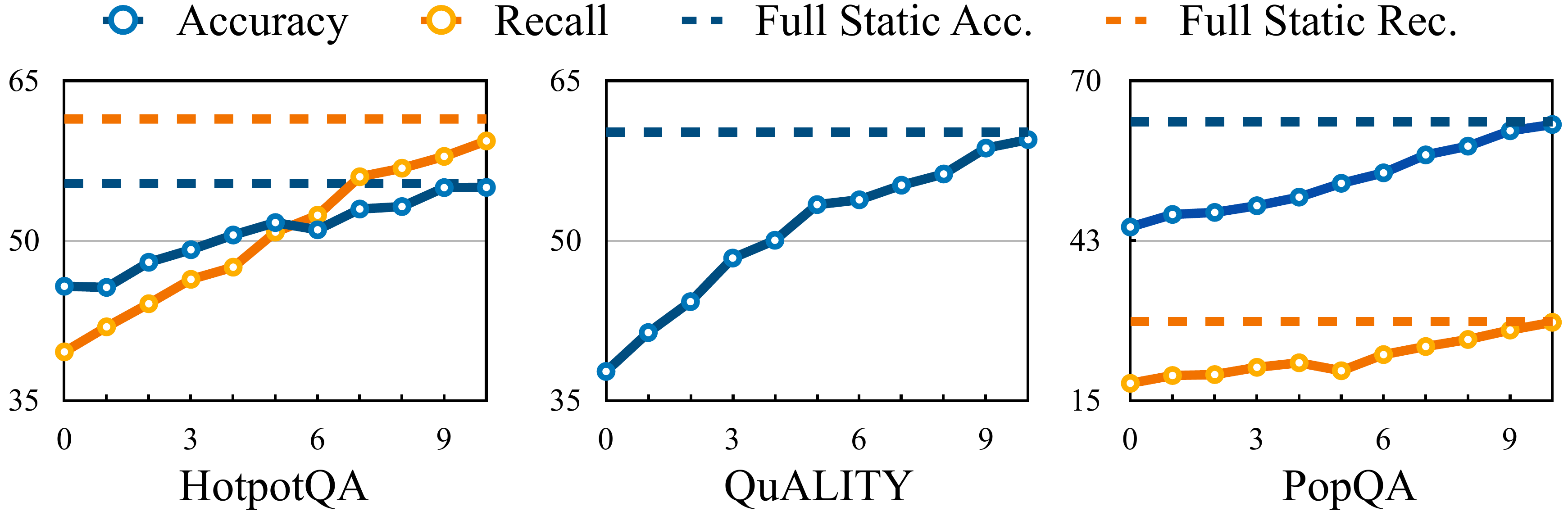}
  \caption{EraRAG performance over incremental insertion. Dotted lines represent static full-graph performance, while solid lines show EraRAG's incremental performance as new data is inserted.}
  \label{fig:perf}
\end{figure*}

\section{Discussions}

  

\begin{table}[t]
  \centering
  \caption{Abstract QA results: EraRAG vs. GraphRAG (top) and EraRAG vs. RAPTOR (bottom).}
  \label{abstract}
  \scalebox{1.2}{  
  \begin{tabular}{c|cccc}
    \toprule
    Dataset & Comp. & Div. & Emp. & Overall \\
    \midrule
    Mix & 56\% & 52\% & 49\% & 51\% \\
    CS & 53\% & 58\% & 48\% & 55\% \\
    Legal & 33\% & 54\% & 42\% & 42\% \\
    MultiSum & 56\% & 42\% & 55\% & 52\% \\
    \bottomrule
  \end{tabular}
  }

  \vspace{2mm}
  \scalebox{1.2}{ 
  \begin{tabular}{c|cccc}
    \toprule
    Dataset & Comp. & Div. & Emp. & Overall \\
    \midrule
    Mix & 67\% & 62\% & 47\% & 54\% \\
    CS & 52\% & 48\% & 41\% & 46\% \\
    Legal & 58\% & 42\% & 61\% & 52\% \\
    MultiSum & 51\% & 57\% & 53\% & 53\% \\
    \bottomrule
  \end{tabular}
  }
\end{table}

\begin{table*}[t]
  \centering
      \caption{Query performance of EraRAG with different initial graph coverage.}
  \label{segmen}
  \scalebox{1.2}{
  \begin{tabular}{c|ccccccccccc}
  \toprule
    Performance & 0\% & 10\% & 20\% & 30\% & 40\% & 50\% & 60\% & 70\% & 80\% & 90\% & 100\%\\
    \midrule
    Accuracy & 41.3 & 45.9 & 53.9 & 58.2 & 61.1 & 62.3 & 62.1 & 62.0 & 62.5 & 61.4 & 62.9 \\
    Recall & 13.9 & 14.0 & 22.9 & 32.8 & 36.6 & 39.3 & 40.0 & 42.0 & 42.8 & 42.3 & 42.9 \\
    \bottomrule
  \end{tabular}
  }
\end{table*}

\noindent\begin{tikzpicture}
\filldraw (0,0) -- (-0.15,0.08) -- (-0.15,-0.08) -- cycle ; 
\end{tikzpicture} \textbf{Exp-1: Efficiency Analysis under Small-Scale Incremental Insertions.} In the dynamic insertion consumption experiments, we conducted ten consecutive insertions and measured the total graph updating time and token consumption. Although each insertion accounted for only 5\% of the total corpus, the absolute data size was still considerable given the scale of the dataset. To further evaluate the performance of \texttt{EraRAG} and baseline methods under more fine-grained, small-scale incremental insertions, we conducted an additional experiment on the MultihopRAG dataset. Specifically, we first constructed the initial graph using 50\% of the entire corpus, followed by a single insertion consisting of one entry, which was segmented into two chunks. We recorded the graph average update time and token consumption and analyzed the results.

The results presented in Figure~\ref{fig:onechunk} show that the advantage of \texttt{EraRAG} becomes more pronounced in small-scale updates. \texttt{EraRAG} completes the update in approximately 20 seconds, whereas baseline methods require significantly more time. Compared to the RAPTOR and HippoRAG methods, EraRAG achieves more than an order of magnitude reduction in both update time and token cost. When compared to the GraphRAG method, EraRAG demonstrates a two-order-of-magnitude reduction in update overhead. These findings demonstrate that \texttt{EraRAG} is well-suited for real-world scenarios involving continuously evolving corpora, offering efficient handling of both large-scale and fine-grained incremental updates.

\begin{figure}[t]
  \centering
  \includegraphics[width=0.95\linewidth]{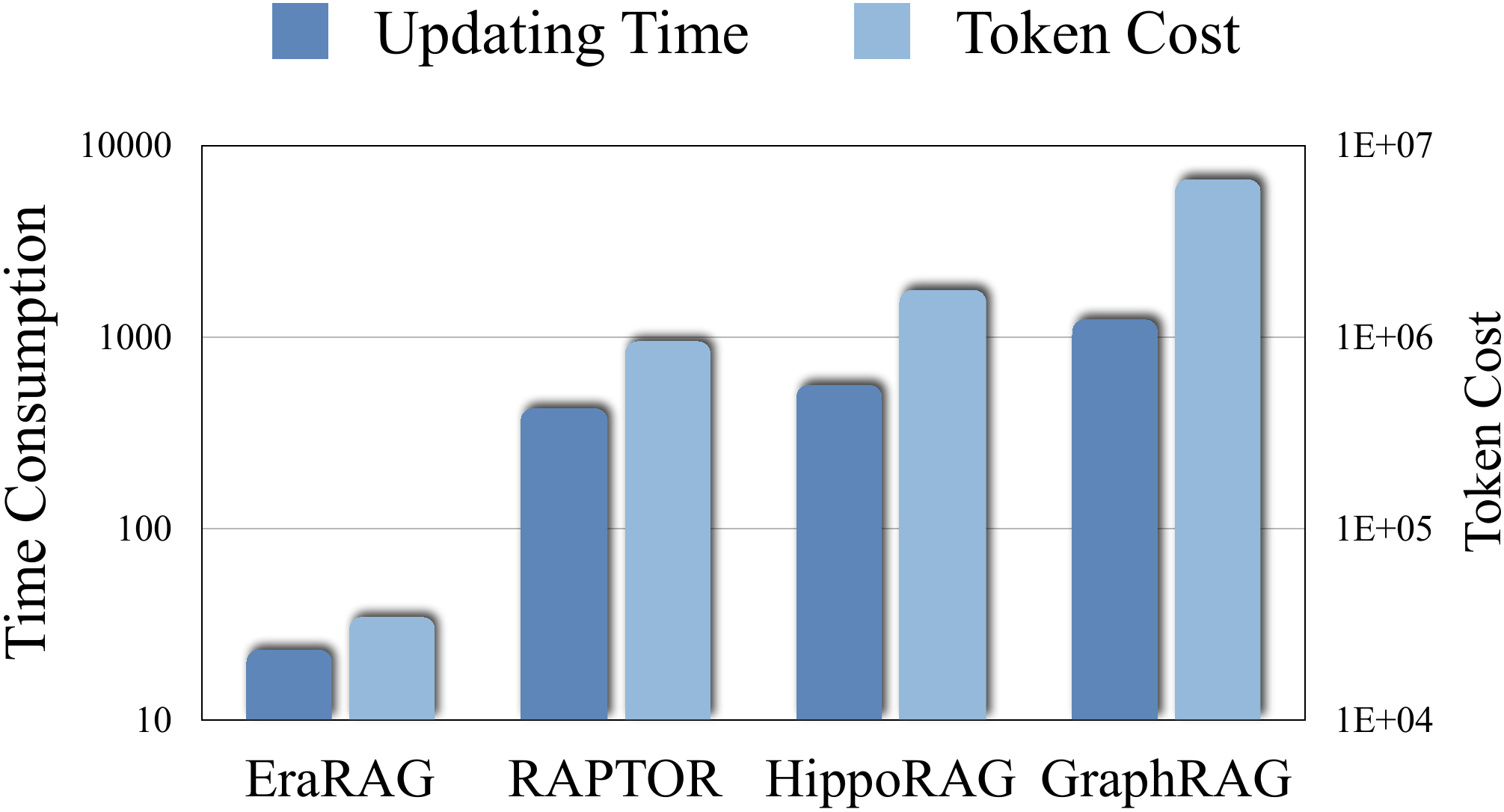}
    \caption{Token consumption and graph updating time under small-scale incremental insertions.}
  \label{fig:onechunk}
\end{figure}

\noindent\begin{tikzpicture}
\filldraw (0,0) -- (-0.15,0.08) -- (-0.15,-0.08) -- cycle ; 
\end{tikzpicture} \textbf{Exp-2: Abstract QA Performance of \texttt{EraRAG}.}
Aside from specific QA tasks, we conducted tests to evaluate \texttt{EraRAG}'s performance on abstract queries. For the abstract QA tasks, prior work~\cite{grag,lrag} proposed a LLM-guided evaluation method, where a LLM evaluator is utilized to evaluate the performance of two models based on comprehensiveness, diversity, empowerment, and a final overall result, which will be adapted in this section. The result of the evaluation will be displayed in a head-to-head win rate percentage form judged by a prompted LLM. We conduct the experiments against GraphRAG and RAPTOR on the well-tested abstract dataset UltraDomain~\cite{ultra} in domains of computer science, legal, and mixed knowledge. We also employ the abstract summary problems offered by Multihop-RAG, known as MultihopSum~\cite{multihop}. 

The experimental results are displayed in Table \ref{abstract}. Compared to GraphRAG, our model consistently achieves higher scores in comprehensiveness, diversity, and empowerment, with particularly strong gains on the CS and Legal domains. Against RAPTOR, our method also demonstrates superior performance across all metrics, notably improving results on the Mix and MultiSum datasets. These results highlight the robustness and generalizing ability of \texttt{EraRAG} for abstract query generation, outperforming both baselines in consistency and overall quality.

\noindent\begin{tikzpicture}
\filldraw (0,0) -- (-0.15,0.08) -- (-0.15,-0.08) -- cycle ; 
\end{tikzpicture} \textbf{Exp-3: Effect of Initial Graph Coverage on Retrieval Performance.} While \texttt{EraRAG} enables dynamic corpus expansion via selective re-construction, the quality of the final retrieval graph may still depend on the initial graph coverage. Sparse initialization can lead to structural noise or poor semantic representation, undermining the effectiveness of later insertions. This experiment quantifies how varying the initial graph ratio impacts final retrieval performance after all data is incorporated. We vary the initial coverage from 0\% to 100\%, inserting the remaining data incrementally using \texttt{EraRAG}. After all data is added, we evaluate the final Accuracy and Recall.
We perform this experiment using the MultihopQA dataset and the results are shown in Table~\ref{segmen}. We can see that both metrics improve with larger initial graphs. Recall grows quickly at low coverage, reflecting early semantic scaffolding, but continues to improve gradually throughout. Accuracy saturates around 50\%, indicating that a well-formed backbone graph enables precise retrieval, while small initial graphs cause structural drift that affects answer quality. These results confirm that the initial structure has a lasting impact, and that 50–70\% coverage offers a trade-off between performance and flexibility.

\noindent\begin{tikzpicture}
\filldraw (0,0) -- (-0.15,0.08) -- (-0.15,-0.08) -- cycle ; 
\end{tikzpicture} \textbf{Exp-4: Effect of Segment Size on Trade-offs in Structure and Efficiency.} \ifarxiv
\else
In \texttt{EraRAG}, ensuring that buckets are partitioned into appropriately sized segments is critical for constructing a well-structured and efficient retrieval graph.
 We investigate how varying the segment size tolerance $\delta$—with a fixed average length $\hat{c}$ and bounds $\hat{c} \pm \delta$—affects segmentation behavior and retrieval performance. Intuitively, a larger $\delta$ better preserves semantic boundaries but may yield uneven abstraction across nodes, harming graph coherence. Smaller $\delta$ enforces uniformity, aiding hierarchy but introducing unnecessary splits and merges that increase token cost and may degrade retrieval.

We evaluate five scaled thresholds: $0.5\cdot\delta$, $0.75\cdot\delta$, $\delta$, $1.5\cdot\delta$, and $2\cdot\delta$, using 50\% of QuALITY for initial graph construction followed by 10 incremental insertions. We record token usage, graph build time, and retrieval accuracy. As shown in Table~\ref{segment}, moderate tightening (e.g., $0.75\cdot\delta$) improves accuracy, confirming that uniform segmentation benefits retrieval, though at increased cost. Larger $\delta$ fails to reduce overhead and causes accuracy to drop, as loose segmentation leads to uneven abstraction and costly subgraph updates. This highlights segmentation’s impact: while splits and merges are cheap, their induced re-summarization dominates cost. Proper bounds thus help balance efficiency and quality.
\fi

\begin{table}
  \caption{Accuracy, token cost and graph building time of EraRAG with different turbulence thresholds on QuALITY.} \label{segment}
  \centering
  \scalebox{1.22}{
  \begin{tabular}{c|ccc}
  \toprule
    Threshold & Accuracy & Tokens & Rebuilding Time\\
    \midrule
    $0.5 \cdot \delta$ & 58.74 & 9.90M & 1025.03s \\
    $0.75\cdot \delta$ & 59.53 & 8.92M & 923.58s \\
    $\delta$ & 59.49 & 8.27M & 878.23 \\
    $1.5\cdot \delta$ & 58.32 & 8.25M & 851.09s \\
    $2\cdot \delta$ & 56.03 & 8.53M & 902.17s \\
    \bottomrule
  \end{tabular}
  }
\end{table}

\noindent\begin{tikzpicture}
\filldraw (0,0) -- (-0.15,0.08) -- (-0.15,-0.08) -- cycle ; 
\end{tikzpicture} \textbf{Exp-5: Robustness Across Backbone Language Models.} \ifarxiv
\else
To evaluate the robustness of \texttt{EraRAG} across different backbone models, we conduct an experiment on the MultihopRAG dataset by replacing the original LLM with \texttt{GPT3.5 turbo}~\cite{3.5} and \texttt{GPT4-o-mini}, both widely used in RAG benchmarks. A one-time insertion of the full corpus is performed under consistent hyperparameters, and we record the graph construction time, token consumption, and F1 score.

The results shown in Table~\ref{backbone} indicate that \texttt{EraRAG} maintains stable performance across models, with all metrics fluctuating within 10\%. A drop in F1 score is observed with \texttt{GPT3.5 turbo}, mainly due to reduced recall. This may stem from \texttt{GPT3.5 turbo}'s general-purpose design, whereas the original \texttt{LLaMA} backbone is more specialized for retrieval and reasoning, benefiting detail-heavy datasets like MultihopRAG. The increased token usage and graph-building time also correlate with this shift. Despite these differences, \texttt{EraRAG} still remains robust and flexible across backbones, supporting effective deployment in diverse environments.
\fi

\begin{table}
  \caption{F1 score, token cost, and graph building time of EraRAG with different backbone LLM.} \label{backbone}
  \centering
  \scalebox{1.22}{
  \begin{tabular}{c|ccc}
  \toprule
    Threshold & F1 & Token & Building Time\\
    \midrule
    Original & 51.03 & 1.87M & 102.3s \\
    \texttt{GPT3.5 trubo} & 49.57 & 1.96M & 112.9s \\
    \texttt{GPT4-o-mini} & 52.21 & 1.90M & 108.7s \\
    \bottomrule
  \end{tabular}
  }
\end{table}

\noindent\begin{tikzpicture}
\filldraw (0,0) -- (-0.15,0.08) -- (-0.15,-0.08) -- cycle ; 
\end{tikzpicture} \textbf{Exp-6: Case study on the correctness of \texttt{EraRAG}.} \ifarxiv
\else
To qualitatively assess the retrieval quality and robustness of \texttt{EraRAG} in incremental settings, we evaluate it on thematic multiple-choice questions based on the full text of \emph{The Wizard of Oz}. Using the same setup as in the main experiment (50\% initial graph + ten increments), we test two query types: \emph{detailed queries} targeting specific facts, and \emph{summary queries} requiring broader contextual understanding. Retrieved chunks and corresponding LLM outputs are analyzed, with representative examples shown in Figure~\ref{fig:detail}.

Results show that \texttt{EraRAG} handles both query types effectively. For detailed queries, it retrieves relevant leaf nodes to extract precise facts. For summary queries, it selects upper-layer summary nodes that, while omitting some specifics (e.g., the slipper color), captures the core narrative, aiding accurate LLM responses. This demonstrates the effectiveness of our hierarchical grouping mechanism in aggregating related content. Moreover, no hallucinations are introduced during incremental updates, confirming the model’s robustness against evolving corpus. This illustrative example also further validates the effectiveness of the proposed customized retrieval mechanism.

\fi

On the other hand, we are also interested in the time distribution across different stages of each update, so we recorded the time spent on each procedure during one update. As shown in Figure~\ref{fig:time}, it is evident that re-summarization dominates the time distribution at all upper levels. Since no summary is generated at layer zero, the embedding update takes the majority of the time. Notably, the procedures other than summarization consume negligible time. This observation suggests that if \texttt{EraRAG} is distributed across localized small models, overall time consumption can be further reduced.

\begin{figure}[t]
  \centering
  \includegraphics[width=\linewidth]{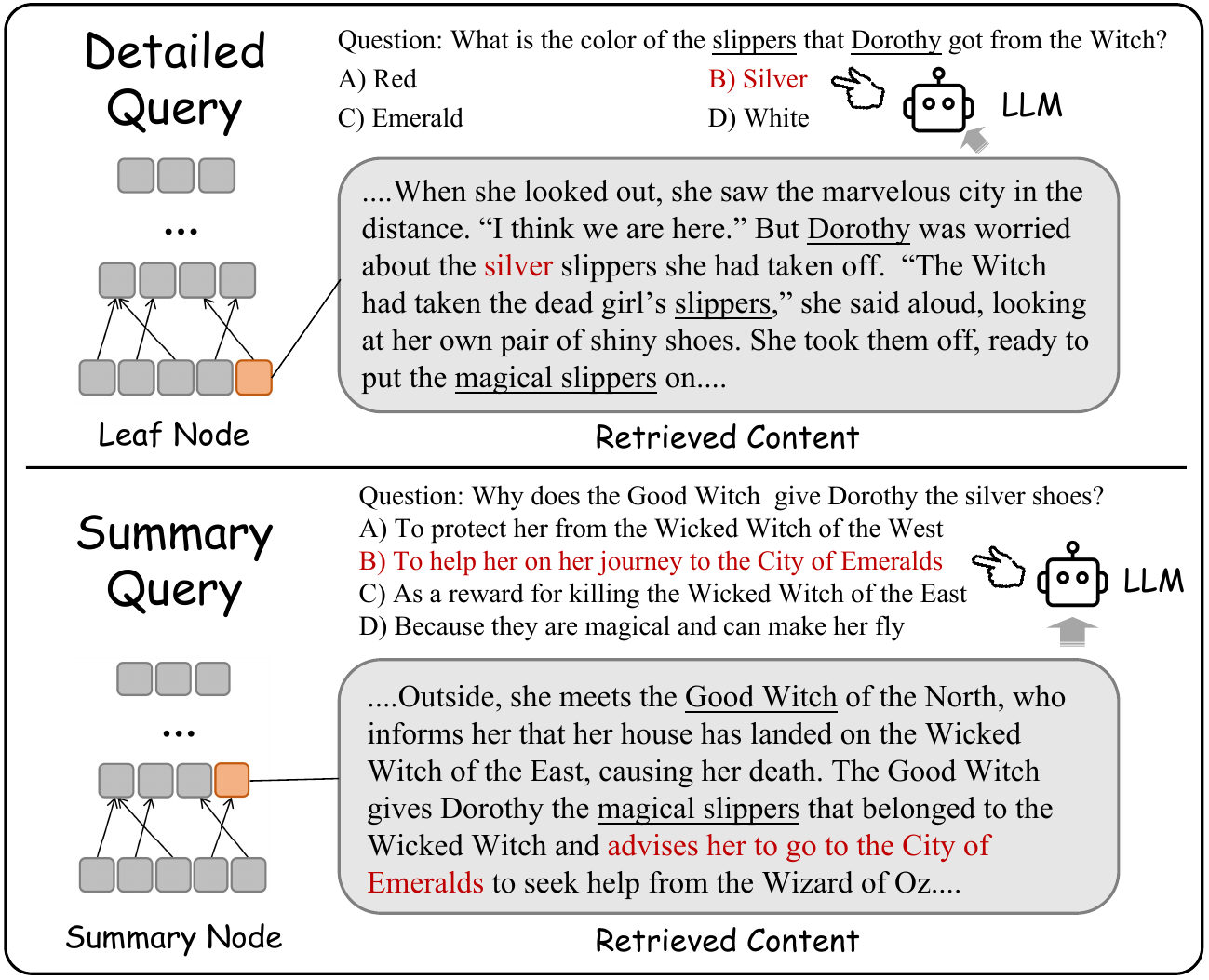}
    \caption{\textbf{Detailed Retrieval of \texttt{EraRAG}.} The \textcolor{darkred}{velvet} colored options are the correct ones. For detailed queries \underline{(top)}, the retrieval process targets leaf node chunks, which contain the original corpus chunks with in-depth information. For summary queries \underline{(bottom)}, summary node chunks are retrieved, utilizing generalized information from multiple paragraphs to address cross-paragraph queries. Enhanced with \texttt{Erarag}, the LLM answers correctly in both cases.}

  \label{fig:detail}
\end{figure}

\noindent\begin{tikzpicture}
\filldraw (0,0) -- (-0.15,0.08) -- (-0.15,-0.08) -- cycle ; 
\end{tikzpicture} \textbf{Exp-7: Effect of Chunk Size on Retrieval Accuracy and Graph Construction Efficiency.} Recent studies have shown that smaller, more focused chunks can improve RAG accuracy at the cost of higher computation~\cite{chunk1}. To assess this trade-off in \texttt{EraRAG}, we evaluate F1 score and graph-building time under varying chunk sizes via a one-time insertion on the MultihopRAG dataset. Results shown in Figure~\ref{fig:chunk} indicate that retrieval quality remains stable across chunk sizes, with F1 fluctuations within 5\%. Surprisingly, smaller chunks do not improve performance but increase graph-building time by 25\%, confirming their higher computational cost. Conversely, larger chunks slightly increase, rather than reduce, build time. Since embedding accounts for less than 1\% of total time (see Exp-6), the delay likely arises from longer LLM summarization for larger chunks. These findings suggest that chunk size has a limited impact on retrieval quality, but careful selection can optimize efficiency without compromising performance.

\begin{figure}[t]
  \centering
  \includegraphics[width=0.9\linewidth]{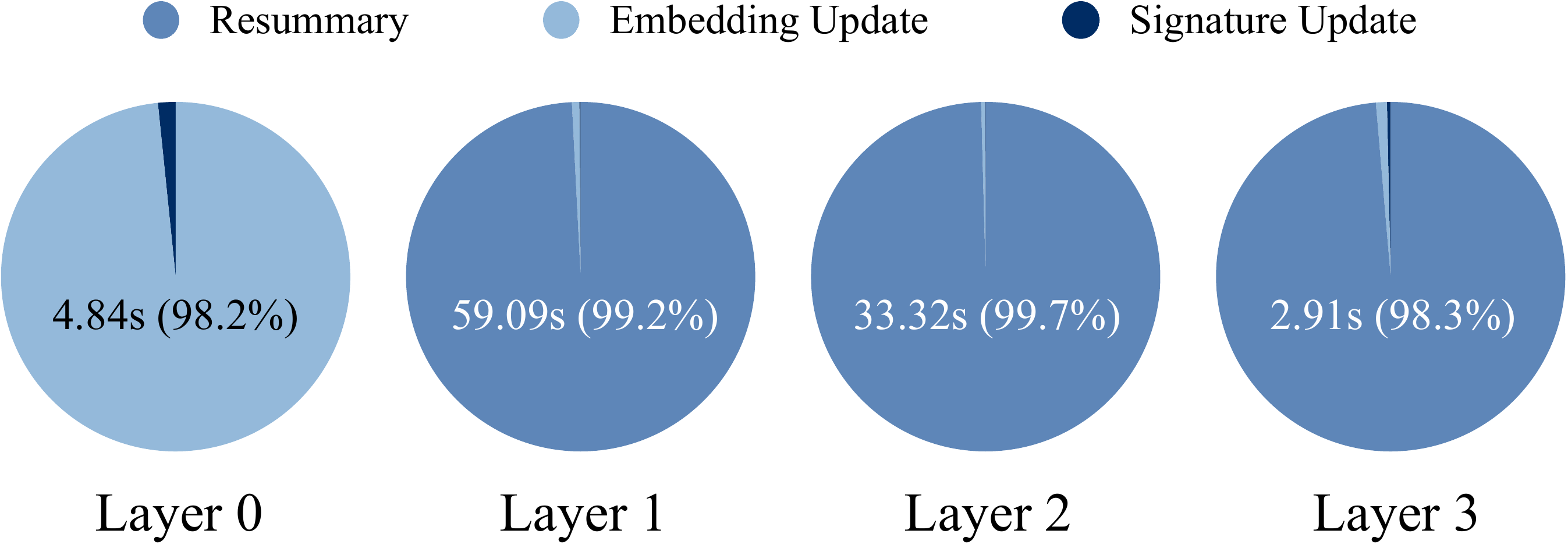}
    \caption{Time consumption of each procedure in graph re-construction for evolving corpora.}
  \label{fig:time}
\end{figure}

\begin{figure}[t]
  \centering
  \includegraphics[width=0.9\linewidth]{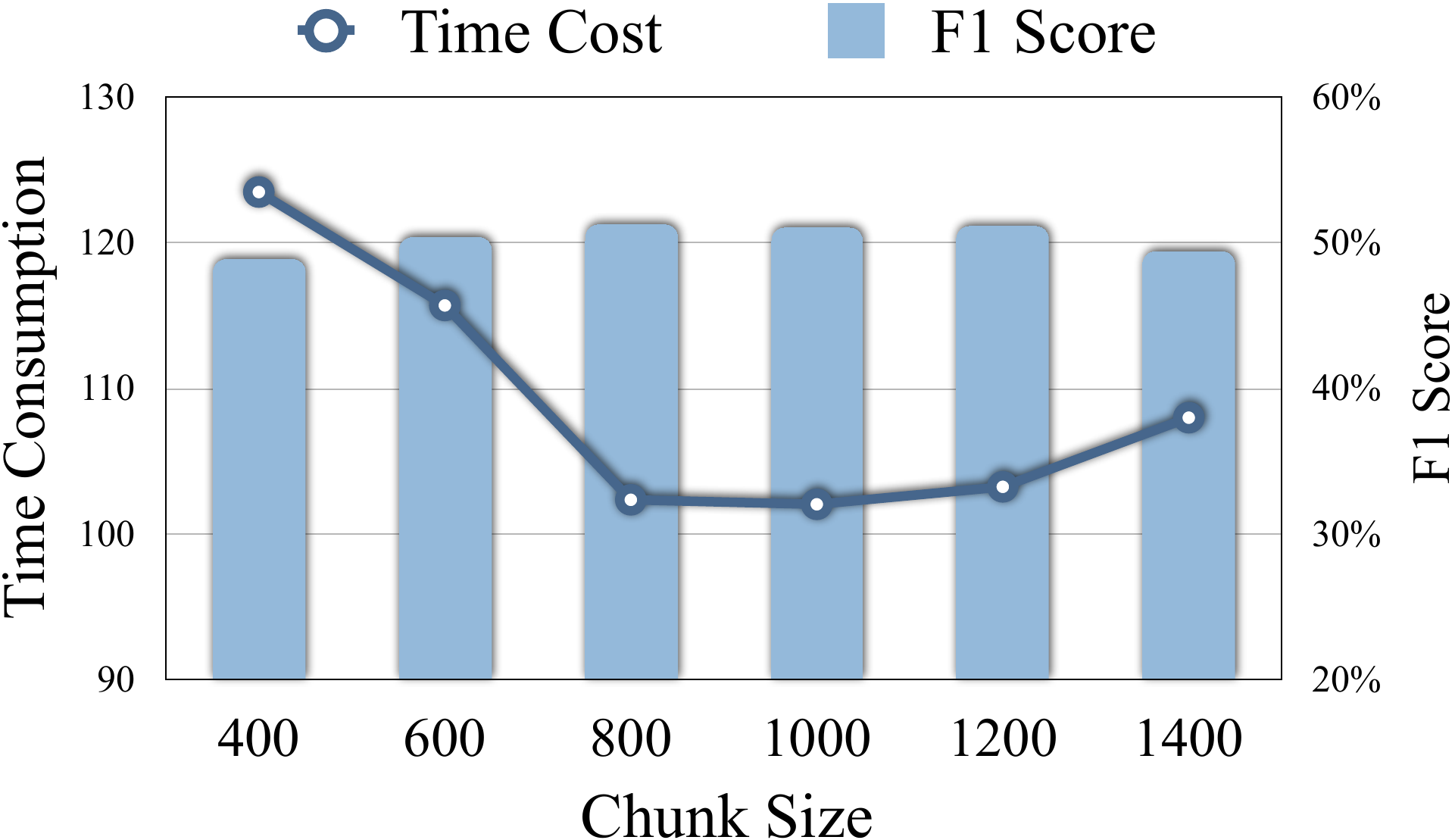}
  \caption{Building time and F1 score over different \textbf{chunk size}.}
  \label{fig:chunk}
\end{figure}

\section{Conclusion}
This paper presents \texttt{EraRAG}, a scalable and efficient graph-based retrieval-augmented generation framework designed to support dynamic corpus updates without full reconstruction. By leveraging a hyperplane-based locality-sensitive hashing mechanism and selective re-clustering, \texttt{EraRAG} enables incremental graph construction while preserving the semantic integrity of previously established structures. Extensive experiments across five QA benchmarks demonstrate that \texttt{EraRAG} consistently outperforms existing RAG baselines in both Accuracy and Recall. Furthermore, our analysis shows that \texttt{EraRAG} significantly reduces token consumption and graph construction time during updates, achieving up to 57.6\% and 77.5\% savings over the next best baseline, respectively. Additional incremental evaluation confirms that \texttt{EraRAG} maintains retrieval quality on par with fully rebuilt graphs, validating its robustness under continual data growth. Overall, EraRAG offers a principled and practical solution for efficient, structure-preserving retrieval in dynamic real-world applications.

\clearpage

\footnotesize
\balance
\bibliographystyle{abbrv}  




\end{document}